\documentclass[12pt]{article}
\usepackage{amsthm}
\usepackage{amsmath,amssymb,setspace}
\title{\large \textbf{Conditional and Unique Coloring of Graphs}}
\author{
\small P.Venkata Subba Reddy  and K. Viswanathan Iyer \thanks {Author for correspondence} \\
\small Dept. of Computer Science and Engineering \\  
\small National Institute of Technology \\ 
\small Tiruchirapalli 620 015, India  \\
\small email : venkatpalagiri@gmail.com, kvi@nitt.edu
} 
\begin{document}
\maketitle
\begin{abstract}  
For integers $k>0$ and $0<r \leq \Delta$ (where $r \leq k$), a conditional $(k,r)$-coloring of a graph $G$ is a proper $k$-coloring of the vertices of $G$ such that every vertex $v$ of degree $d(v)$ in $G$ is adjacent to vertices with at least $\min\{r, d(v)\}$ differently colored neighbors. The smallest integer $k$ for which a graph $G$ has a conditional $(k,r)$-coloring is called the $r$th order conditional chromatic number, denoted by $\chi_r(G)$. For different values of $r$ we first give results (exact values or bounds for $\chi_r(G)$ depending on $r$) related to the conditional coloring of graphs. Then we obtain $\chi_r(G)$ of certain parameterized graphs viz., windmill graph, line graph of windmill graph, middle graph of friendship graph, middle graph of a cycle, line graph of friendship graph, middle graph of complete $k$-partite graph, middle graph of a bipartite graph and gear graph. Finally we introduce \emph{unique conditional colorability} and give some related results. \medskip 
\end{abstract}
\section{Introduction}
Let $G= (V(G),E(G))$ be a simple, connected, undirected graph. For a vertex $v \in V(G)$, the \textit{open neighborhood} of $v$ in $G$ is defined by $N_G(v)$= \{$u \in V(G):(u,v) \in E(G)$\}, and the degree of $v$ is denoted by $d(v)$=$|N_G(v)|$. Let $\delta(G),\;\Delta(G)$ and $\omega(G)$ (or simply $\delta,\;\Delta$ and $\omega$) denote, respectively the minimum degree, the maximum degree and the clique number of $G$. For an integer $k>0$, a \textit{proper} $k$-coloring of a graph $G$ is a surjective mapping $c \colon V(G) \to \{1,\ldots,k \}$ such that if $(u,v) \in E(G)$, then $c(u) \neq c(v)$. The smallest $k$ such that $G$ has a proper $k$-coloring is the \textit{chromatic number} $\chi(G)$  of $G$. Given a set $S \subseteq V(G)$ we define $c(S)= \{c(u) : u \in S$ \}. For integers $k>0$ and $0<r \leq \Delta$ (where $r \leq k$), a \textit{conditional} $(k,r)$-coloring of $G$ is a surjective mapping $c \colon V(G) \to \{1,\ldots,k \}$ such that both the following conditions (C1) and (C2) hold:
\begin{quote}
(C1) If $(u,v) \in E(G)$, then $c(u) \neq c(v)$. \\
(C2) For any $v \in V(G)$, $|c(N_G(v))| \geq $ min \{$d(v),r$ \}. 
\end{quote}
The smallest integer $k$ such that $G$ has a conditional $(k,r)$-coloring is called the $r$th order conditional chromatic number of $G$, denoted by $\chi_r(G)$. It is proved in \cite{Li09}, that deciding whether $\chi_r(G) \leq k$ is $NP$-complete. For undefined notations/terminology we refer to standard texts in graph theory such as \cite{Bal,Gol80,Har69,West03}). 
\section{Conditional colorability of some graphs}	
We start with the definitions of Cartesian product and join of two graphs. Let $G=(V(G),E(G))$ and $H=(V(H),E(H))$ be two graphs. The Cartesian product $G \Box H$ of $G$ and $H$ has the vertex set $V(G) \times V(H)$ and the edge set $E(G \Box H)=\{((x_1,x_2),(y_1,y_2)): (x_1,y_1) \in E(G)\; \text{and} \; x_2=y_2, \; \text{\textit{or}} \; (x_2,y_2) \in E(H) \; \text{and} \; x_1=y_1  \}$. The join $G + H$ has the vertex set $V(G) \cup V(H)$ and the edge set $E(G) \cup E(H) \cup CE$ where the cross-edge set $CE=\{(u_g,u_h): u_g \in V(G) \; \text{and} \; u_h \in V(H) \}$.
\newtheorem{thms1}{Theorem}[section] 
\begin{thms1}
Let $G_1$ and $G_2$ be two graphs where $\chi(G_1)=k_1$, $\chi(G_2)=k_2$ and w.l.o.g. let $k_1 \le k_2$. Then $\chi_r (G_1 + G_2)$ = $\chi (G_1+ G_2)$ = $k_1 + k_2$, where  $r \leq k_1+ 1$.
\end{thms1}
\begin{proof}
In  the graph $G_1+ G_2$, $V(G_2) \subset N_{G_1+ G_2}(u)$ if $u \in V(G_1)$ or $V(G_1) \subset N_{G_1+ G_2}(u)$ if  $u \in V(G_2)$. Therefore $c(V(G_1)) \cap c(V(G_2)) = \emptyset $, and in any proper $k$-coloring of $G_1+ G_2$, for all $u \in V(G_1+ G_2)$, $|c(N_{G_1+G_2}(u))| \geq \min \{d(u),r \} $. This implies that  every proper $k$-coloring of $G_1+ G_2$ is also a proper $(k,r)$-coloring -- if not, we get a contradiction:  suppose that $\chi(G_1 + G_2) = k \neq k_1 + k_2$; since $c(V(G_1)) \cap c(V(G_2)) = \emptyset $, either $k_1 \neq \chi(G_1)$  or  $k_2 \neq \chi(G_2)$ which contradicts the given condition.  
\end{proof}
\newtheorem{thms3}[thms1]{Theorem}  
\begin{thms3}
Let $T_1, T_2$ be two non trivial trees with $n_1,n_2$ number of vertices respectively and w.l.o.g. let $n_1 \leq n_2$. Then $\chi_r(T_1+T_2)= 2(r-1)$, where $4 \leq r \leq \  n_1+1$. 
\end{thms3}
\begin{proof}
Every nontrivial tree has at least two vertices with degree one \cite{Har69}. Therefore there exist vertices $u,v$ where $u \in V(T_1)$ and $v \in V(T_2)$, such that $d(u)=d(v)=1$. Therefore, $d_u(T_1+T_2)= 1+n_2$ and $d_v(T_1+T_2)= 1+n_1$. If $\chi_r(T_1+T_2) < 2(r-1)$, then either $|c(V(T_1))| < r-1$ or $|c(V(T_2))| < r-1$ or both because $c(V(T_1)) \cap c(V(T_2)) = \emptyset $. Hence (C2) is violated at $u$ or $v$ or both. Therefore, $\chi_r(T_1+T_2) \geq 2(r-1)$. Since every tree is $2$-colorable and $r \geq 4 $, properly color $V(T_1)$, $V(T_2)$ in $T_1+T_2$ using $r-1$ colors each such that $|c(V(T_1+T_2))|=2(r-1)$. The resulting coloring is  a conditional $(2(r-1),r)$-coloring of $T_1+T_2$, as (C1) is satisfied because $c(V(T_1)) \cap c(V(T_2)) = \emptyset $, and (C2) is satisfied because $c(V(T_1)) \cap \  c(V(T_2)) = \emptyset$ and for all $w \in V(T_1+T_2)$, $|c(N_{T_1+T_2}(w))| \geq (r-1)+1 \geq$  min  $\{d(w),r \}$. Hence the result. 
\end{proof}
\newtheorem{thms4}[thms1]{Theorem}  
\begin{thms4}
Given any two graphs $G_1$ and $G_2$, let $r_1$ and $r_2$ be such that $r_1 \geq \delta(G_1)$ and $r_2 \geq \delta(G_2)$.  Then  $\chi_r(G_1 \ \Box \ G_2) \leq \chi_{r_1}(G_1). \chi_{r_2}(G_2)$  where $r \leq \delta(G_1)+ \delta(G_2)$.
\end{thms4}
\begin{proof}
Let $\chi_{r_1}(G_1)= g_1$ and $\chi_{r_2}(G_2)= g_2$. Let $c_{G_1}$ (resp. $c_{G_2}$) be a proper $(g_1,r_1)$- (resp. $(g_2,r_2)$-) coloring of $G_1$ (resp. $G_2$). Then let $c_{G_1 \Box G_2}$ be a coloring of $G_1 \Box G_2$ wherein we assign to any vertex $(u_1,u_2) \in V(G_1 \ \Box \ G_2)$ the color denoted by the ordered pair $(c_{g_1}(u_1),c_{g_2}(u_2))$. This coloring uses $g_1.g_2$ colors and it defines a proper coloring of $G_1 \ \Box \ G_2$. Therefore $c_{G_1 \Box G_2}$ satisfies (C1).  
Let $(u_1,u_2) \in V(G_1 \ \Box \ G_2)$ such that $u_1 \in V(G_1)$ and $u_2 \in V(G_2)$. Since $c_{G_1}$ and $c_{G_2}$ satisfy (C2), by the definition of $G_1 \Box G_2$, a vertex $(u_1,u_2)$ has at least $\min \{r_1, \delta(G_1)\}=\delta(G_1)$ distinctly colored neighbors of the form $(u',u_2)$ because $|c(N_{G_1}(u_1))| \geq  \delta(G_1)$ and at least $\min \{r_2, \delta(G_2)\}=\delta(G_2)$ distinctly colored neighbors of the form $(u_1,u'')$  because  $|c(N_{G_2}(u_2))| \geq  \delta(G_2)$. Therefore $|c(N_{G_1 \ \Box \ G_2}((u_1,u_2))| \geq \delta(G_1)+ \delta(G_2)\geq r$. Hence $c_{G_1 \Box G_2}$ satisfies (C2) and the result follows.   
\end{proof}
\newtheorem{thms2}[thms1]{Theorem}  
\begin{thms2}
Let $G(V_1,V_2,E)$ be a bipartite graph, $S_1=\bigcap_{u \in V_1}N_G(u)$, $S_2=\bigcap_{v \in V_2}N_G(v)$ and w.l.o.g. let $|S_1| \le |S_2|$. Then  $\chi_r(G) = 2r$ where $r \leq |S_1|$ . 
\end{thms2} 
\begin{proof}
In any proper coloring of $G$, from the given conditions $|c(V_1)| \geq r$ and $|c(V_2)| \geq r$ as $G$ is bipartite. Since $r \leq |S_1|$ and $c(S_1) \cap c(S_2) = \emptyset $ we have $\chi_r(G) \geq 2r$. But there exists a proper $2r$-coloring of $G$ such that $|c(S_1)|= |c(S_2)|=r$ because every bipartite graph is bicolorable and $r \geq 2$. This coloring also satisfies (C2) as $S_1 \subseteq V_2$ and $S_2 \subseteq V_1$. Thus  $\chi_r(G) \leq 2r$. Hence $\chi_r(G) = 2r$.
\end{proof}
\newtheorem{thms5}[thms1]{Theorem}  
\begin{thms5}
Let $L(T)$ be the line graph of complete $k$-ary tree $T$ with height $h \geq 2$. Then \[\chi_r(L(T)) = \left\{ 
\begin{array}{l l}
k+1, & \quad \mbox{if $r \leq k$. \textsl{}}\\  
2k+1,   & \quad \mbox {if $ r = \Delta$. \textsl{}}\\ \end{array} \right. \] 
\end{thms5}
\begin{proof}
Let $V(L(T))= \left \{v_1,v_2,\dotsc,v_{e(h)}\right \}$, where $e(h)=\frac{k^{h+1}-1}{k-1}-1$. In $T$  we assume that the root is at level $0$ and for each $l$ ($1 \leq l \leq h$), $v_{e(l-1)+1}$ to $v_{e(l)}$ represent the edges between levels $l-1$ and $l$, numbered from `left' to `right'. It can be  seen that $\Delta(L(T))=2k$ and $\omega(L(T))=k+1$. In the ordering $v_{e(h)}, \dotsc, v_1$ of the vertices of $L(T)$, for each $i$ ($1 \leq i \leq e(h)$), $v_i$ is a simplicial vertex in the subgraph induced by $\{v_i,\dotsc,v_1\}$. Hence the ordering is a p.e.o. and $L(T)$ is chordal. As every chordal graph is perfect,  $\chi(L(T))= \omega(L(T))= k+1$. Since every vertex of $L(T)$ is in a $K_{k+1}$, we also have $\chi_r(L(T))=k+1$ if $r \leq k$. Thus $\chi_r(L(T)) = k+1$, if $r \le k$.  From \cite{Lai06} we know $\chi_r(G) \geq \min \{r, \Delta \}+1$. Taking $G=L(T)$ we have $\chi_r(L(T)) \geq \min \{r, \Delta \}+1=2k+1$ if $r=\Delta$. Similar to the greedy (vertex) coloring, color the vertices in the order $v_1, \dotsc, v_{e(h)}$ by assigning to each vertex the first available color not already used for any of the lower indexed vertices within distance two. In the assumed order, each vertex has at most $\Delta$ lower indexed vertices within distance two; therefore $\chi_\Delta(L(T)) \leq \Delta +1=2k+1$.  Hence  $\chi_r(L(T)) = 2k+1$ if $r = \Delta$.  
\end{proof} 
\section{Conditional colorability of some parameterized \\ graphs}
For this section we need the following definitions. The \textit{middle graph } $M(G)$ of $G$  is the graph whose vertex set corresponds to $V(G) \cup E(G)$; in $M(G)$ two vertices are adjacent iff 
$(i)$ they are adjacent edges of $G$ or 
$(ii)$ one is a vertex and the other is an edge incident with it \cite{Mich83}. The \textit{windmill graph} $Wd(k,n)$ consists of $n$ copies of $K_k$ and identifying one vertex from each $K_k$ as the common center vertex. In particular $Wd(3,n)$ is called the \textit{Friendship graph} $F_n$ \cite{Galin09}. A \textit{complete $k$-partite graph}  $K(n_1,\ldots,n_k)$ has vertex set $V=V_1 \cup \ldots \cup V_k$ where $V_1,\ldots,V_k$ are mutually disjoint with $|V_i|=n_i$; each vertex $v \in V_i$ is connected to all vertices of $V \setminus V_i, \; i=1,\ldots,k$. The \textit{$n$-gear $G_n$} consists of a cycle $C_{2n}$ on $2n$ vertices where every other vertex on the cycle is adjacent to a $(2n+1)^{\rm th}$ center vertex labeled $v_0$. The  vertices in the $C_{2n}$ are labeled sequentially $v_1, \ldots,v_{2n}$ such that for  $1 \leq i \leq 2n-1$, $v_i$ is adjacent to $v_{i+1}$, $v_1$ is adjacent to $v_{2n}$,  and every vertex in $V_{oddi}= \{v_i |\; \text{i is odd and}\; 1 \le i \le 2n-1  \}$ is adjacent to $v_0$. \medskip \\
We begin with two lemmas followed by our propositions.
\newtheorem{lem1}{Lemma}[section]
\begin{lem1}
For a non-negative integer $r \leq \Delta$ let \textit{Vset-$d2r$} in a graph $G$ be a set $S_{d2r} \subseteq V(G)$ with the following two properties:
\begin{quote}
$(i)$ For all $u \in S_{d2r}$,  $d(u) \leq r$.  \\
$(ii)$ For all $u_1,u_2 \in S_{d2r}$ either $(u_1,u_2) \in E(G)$ or there exists a $u_3 \in S_{d2r}$ such that $u_1,u_2 \in N(u_3)$ or both.
\end{quote}
Then $\chi_r(G) \geq |S_{d2r}|$.
\end{lem1}
\begin{proof}
Assume that $\chi_r(G) < |S_{d2r}|$. Then there exist at least two vertices $u_1,u_2 \in S_{d2r}$ such that $c(u_1)=c(u_2)$. By the definition of Vset-$d2r$    $(ii)$ holds; if $(u_1,u_2) \in E(G)$, (C1) is violated; if $u_3 \in S_{d2r}$ such that $u_1,u_2 \in N(u_3)$, $|c(N(u_3))|< min\{r,d(u_3)\}=d(u_3)$ and hence (C2) is violated at $u_3$. Therefore  $\chi_r(G) \geq |S_{d2r}|$.
\end{proof}
\newtheorem{lem2}[lem1]{Lemma}
\begin{lem2}
Given a graph $G$, let $c \colon V(G) \to \{1,\ldots,k \}$ be a coloring such that for a given $r$, $c$ satisfies (C2). Let the condition (C3) be
\begin{quote}
(C3) For each edge $uv$ in $G$ there exists a vertex $w$ such that $d(w)$ $\leq r$ and $u,v \in N_{G}(w)$.
\end{quote}
If $G$ satisfies (C3) also then $c$ satisfies (C1) and hence $c$ defines a conditional $(k,r)$-coloring of $G$.
\end{lem2}
\begin{proof}
The proof by contradiction is straightforward.
\end{proof}
\newtheorem{thm1}[lem1]{Proposition} 
\begin{thm1} 
For integers $k \geq 3, n \geq 2$ let $Wd(k,n)$ be a windmill graph. Then 
\[\chi_r(Wd(k,n)) = \left\{ 
\begin{array}{l l}
k, & \quad \mbox{if $2 \leq r \leq k-1$. \textsl{}}\\  
min \{r,\Delta\}+1, & \quad \mbox{if $ r \geq k$. \textsl{}} \\ \end{array} \right. \]
\end{thm1}
\begin{proof}
Every vertex $v$ of $Wd(k,n)$ is contained in a $K_k$; it can be seen that $|c(N(v))| \geq k-1$ in any proper coloring $c$ of $Wd(k,n)$. Therefore if $2 \leq r \leq k-1 $,  conditional $(\chi(Wd(k,n)),r)$-coloring of $Wd(k,n)$ exists and we know that $\chi(Wd(k,n))=k$. 
From \cite{Lai06} we have $\chi_r(G) \geq min \{r, \Delta \}+1$. Taking $G=Wd(k,n)$ we get $\chi_r(Wd(k,n))$  $\geq min \{r, \Delta \}+1$. In $Wd(k,n)$ only the center vertex has a degree $n(k-1)> k$. For a $k'$ if $k' > k$, then every $k$-colorable graph is also $k'$-colorable. Hence if $ r \geq k$, then a proper $(min \{r,\Delta\}+1)$-coloring of $Wd(k,n)$ exists, which is also a conditional $(min \{r,\Delta\}+1,r)$-coloring. Therefore $\chi_r(Wd(k,n)) \leq min \{r, \Delta \}+1$. Hence $\chi_r(Wd(k,n)) = min \{r, \Delta \}+1$  if $r \geq k$.    
\end{proof}
\newtheorem{thm}[lem1]{Proposition} 
\begin{thm}
Let $L(Wd(k,n))$ be the line graph of windmill graph $Wd(k,n)$. Then 
\begin{equation*}
\chi_\Delta(L(Wd(k,n)))= n(k-1)+\binom{k-1}{2}= z \;(say).
\end{equation*} 
\end{thm}
\begin{proof}
It follows that $|V(L(Wd(k,n)))|=n\binom{k}{2} = inx(k,n)$ (say). Let $V(L(Wd(k,n)))= \{v_1,\ldots,v_{inx(k,n)}\}$. We assume that in $Wd(k,n)$, for all $1 \leq i \leq n,   
 v_{(i-1)(k-1)+1}$\  to\  $v_{i(k-1)}$ and $v_{n(k-1)+inx(k-1,i-1)+1}$\  to\  $v_{n(k-1)+inx(k-1,i)}$ represent respectively the edges of $i^{th}$ copy of $K_k$ incident with and not incident with the center vertex. It can be seen that $L(Wd(k,n))$ has a clique $\{ v_1,\ldots,v_{n(k-1)} \}$ and a Vset-$d2r$ $S_{d2r}=\{v_1,\ldots,v_z\}$. By lemma 3.1, $\chi_\Delta(L(Wd(k,n)))$ $\geq |S_{d2r}|=z$. We now define the coloring assignment $c \colon V(L(Wd(k,n))) \to \{1,\ldots,z \}$ as follows: 
\begin{equation*}
c(v_i) =
\begin{cases}
i, & \text{if $1 \leq i \leq z.$}\\
i \mod {\binom{k-1}{2}}+n(k-1)+1, & \text{otherwise.}
\end{cases}
\end{equation*}
In $c$ the \textit{if}-case uses $z$ and the \textit{otherwise}-case doesn't use any extra color. For all $1 \leq i \leq n$ \\ $|c(\{v_{(i-1)(k-1)+1},\ldots,v_{i(k-1)}\})$ $\cup$  $c(\{v_{n(k-1)+ inx(k-1,i-1)+1},\ldots,v_{n(k-1)+ inx(k-1,i)}\})|$ \\ = $|\{v_{(i-1)(k-1)+1},\ldots,v_{i(k-1)}\} \cup \{v_{n(k-1)+inx(k-1,i-1)+1},\ldots,v_{n(k-1)+inx(k-1,i)}\}|$ and \\ $|c\{v_1,\ldots,v_{n(k-1)}\}|=n(k-1)$; hence (C2) is satisfied at all the vertices. The graph $G= L(Wd(k,n))$ can be seen to satisfy (C3) in lemma 3.2. By lemma 3.2, $\chi_\Delta (L(Wd(k,n))) \leq z$; hence $\chi_\Delta (L(Wd(k,n))) = z$.
\end{proof}
\newtheorem{thm2}[lem1]{Proposition} 
\begin{thm2}
Let $L(F_n)$ be the line graph of $F_n$. Then
\begin{equation*}
\chi_r(L(F_n)) =
\begin{cases}
2n, & \text{if $r < \Delta$.}\\ 
2n+1 , & \text{if $r = \Delta$.}
\end{cases}
\end{equation*}
\end{thm2}
\begin{proof}
We have the following two cases: \\
\textbf{Case 1:} $r < \Delta:$ Let $r'=\Delta-1$. Since $|V(L(F_n))|=3n$, let $V(L(F_n))= \{v_1,\ldots,v_{3n}\}$. We assume that for all $1 \leq i \leq n, v_{2i-1}$ and $v_{2i}$ represent the edges of $i^{\rm th}$ copy of $K_3$ incident with the center vertex and $v_{2n+i}$ represents the edge of $i^{\rm th}$ copy of $K_3$ not incident with the center vertex of $F_n$. As $\{v_1,\ldots,v_{2n} \}$ is the maximum size clique of $L(F_n)$, $\omega(L(F_n))=2n$. Since $\chi_r(G) \geq \omega(G)$, taking $r=r'$ and $G=L(F_n)$ we have $\chi_{r'}(L(F_n))\geq 2n$. We now define the coloring assignment $c \colon V(L(F_n)) \to \{1,\ldots,2n\}$ as follows: 
\begin{equation*}
c(v_i) =
\begin{cases}
i, & \text{if $1 \leq i \leq 2n$.}\\
2n, & \text{if $2n+1 \leq i \leq 3n-1$.}\\
1, & \text{if $i = 3n$.}
\end{cases}
\end{equation*}
Note that in $c$ the first case uses $2n$ colors and the remaining cases use no new colors. It can be verified that $c$ defines a conditional $(2n,r')$-coloring of $L(F_n)$. Thus $\chi_{r'} (L(F_n)) \leq 2n$; hence $\chi_{r'} (L(F_n))= 2n$. From \cite{Lai06} it follows that $\omega(G) \leq \chi_{r_1}(G) \leq \chi_{r_2}(G)$ if $r_1 \leq r_2$. Taking $G=L(F_n),r_1=r$ and $r_2=r'$ it follows that $\chi_r(L(F_n))= 2n$. \\
\textbf{Case 2:} $r = \Delta:$ Since $F_n=Wd(3,n)$, by theorem 1 we have $\chi_\Delta(L(Wd(3,n)))=\chi_\Delta(L(F_n))=2n+1$.
\end{proof}
\newtheorem{thm3}[lem1]{Proposition} 
\begin{thm3}
Let $M(K(n_1,\ldots,n_k))$ be the middle graph of $K(n_1,\ldots,n_k)$. Then 
\begin{equation*}
\chi_\Delta(M(K(n_1,\ldots,n_k)))= k+l.
\end{equation*} 
where $n=\sum_{i=1}^k n_i$ and $l=1/2 \sum_{i=1}^k n_i(n-n_i)$.
\end{thm3}
\begin{proof}
We know that $K(n_1,\ldots,n_k)$ has $l$ edges, $|V(M(K(n_1,\ldots,n_k)))|=l+n$. Let $V(M(K(n_1,\ldots,n_k)))=\{v_1,\ldots,v_{l+n}\}$ and $n_0=0$. We assume that $v_1$ to $v_l$ represent the edges and for all $1 \leq i \leq k, v_{l+1+\sum_{j=0}^{i-1} n_j}$ to $v_{l+\sum_{j=0}^i n_j}$ represent the $i^{\rm th}$ partition vertices of $K(n_1,\ldots,n_k)$. Let $r=\Delta,V_e=\{v_1,\ldots,v_l\}$ and $V_v=\{v_{l+1},\ldots,v_{l+n}\}$. It can be easily seen that $M(K(n_1,\ldots,n_k))$ has a Vset-$d2r$ given by \\ $S_{d2r}=V_e \; \cup \; \{v_{l+1},v_{l+n_1+1},v_{l+(n_1+n_2)+1},v_{l+(n_1+n_2+n_3)+1},\ldots,v_{l+(n_1+\ldots+n_{k-1})+1}\}$. \\ Thus by lemma 3.1, $\chi_r(M(K(n_1,\ldots,n_k))) \geq |S_{d2r}|= k+l$. We now define the coloring assignment $c \colon V(M(K(n_1,\ldots,n_k)))$   $\to \{1,\ldots,k+l\}$ as follows: 
\begin{equation*}
c(v_i) =
\begin{cases}
i, & \text{if $1 \leq i \leq l.$}\\ 
l+p, & \text{otherwise, where $p$ is such that $1+\sum_{j=0}^{p-1} n_j \leq i-l \leq \sum_{j=0}^p n_j$.}
\end{cases}
\end{equation*}
In $c$ the first case uses $l$ colors and the remaining case uses $k$ new colors. We have $|c(V_e)|=|V_e|$ and $c(V_e) \cap c(V_v)= \emptyset$; for any two vertices $v_i,v_{i'} \in V_v$ if there exists no $p$ such that $l+1+\sum_{j=0}^{p-1} n_j \leq i,i' \leq l+\sum_{j=0}^i n_j$ then $c(v_i) \neq c(v_{i'})$; hence (C2) is satisfied at all the vertices. Taking $G=M(K(n_1,\ldots,n_k))$ it can be seen that (C3) of lemma 3.2 is satisfied. By lemma 3.2, $\chi_\Delta (M(K(n_1,\ldots,n_k))) \leq k+l$. Hence $\chi_\Delta (M(K(n_1,\ldots,n_k))) = k+l$. 
\end{proof}
\newtheorem{thm4}[lem1]{Proposition} 
\begin{thm4}
For $n \geq 4$, let $M(C_n)$ be the middle graph of $C_n$. Then \[\chi_r(M(C_n)) = \left\{  
\begin{array}{l l}
3, & \quad \mbox{if $r = 2$. \textsl{}}\\  
4, & \quad \mbox {if $ r =3$. \textsl{}}\\ \end{array} \right. \] 
\end{thm4}
\begin{proof}
Let $V(M(C_n))= \{v_1,\ldots,v_{2n}\}$. We assume that $v_1$ to $v_n$ and  $v_{n+1}$ to $v_{2n}$ represent the vertices and edges of $C_n$ respectively where for all $i$ ($1 \leq i \leq n$) $v_{n+i}$ is incident with both $v_{i}$ and $v_{i \; mod  \; n + 1}$. We have the following cases:  \\
\textbf{Case 1:} $r=2:$ Since $r < \Delta,\; \chi_r(M(C_n)) \geq 3$. We define the coloring assignment $c \colon V(M(C_n)) \to \{1,2,3 \}$ thus:\\  
For even $n$ 
\begin{equation*}
c(v_i) =
\begin{cases}
1, & \text{if $1 \leq i \leq n.$}\\
2, & \text{if $n+1 \leq i \leq 2n$ and $i:$ odd.}\\
3, & \text{otherwise.}
\end{cases}
\end{equation*}
For odd $n$  
\begin{equation*}
c(v_i) =
\begin{cases}
1, & \text{if $i=1$ or both $n+1 \leq i \leq 2n$ and $i:$ odd.}\\
2, & \text{if $i=2n$ or $2 \leq i \leq n-1$.}\\
3, & \text{otherwise.}
\end{cases}
\end{equation*}
In both the cases it can be verified that $c$ defines a conditional $(3,r)$-coloring of $M(C_n)$. Thus $\chi_r (M(C_n)) \leq 3$; hence  $\chi_r (M(C_n))=3$. \\
\textbf{Case 2:} $r=3:$ Since $r < \Delta, \; \chi_r(M(C_n)) \geq 4$. We define the coloring assignment $c \colon V(M(C_n)) \to \{1,2,3,4 \}$ as follows:
\begin{equation*}
c(v_i) =
\begin{cases}
1, & \text{if $n+1 \leq i \leq 2n$ and $(i-n):$ even.}\\
2, & \text{if $1 \leq i \leq n$ and $i:$ odd.}\\
3, & \text{if $i=n+1$ or both $4 \leq i \leq n$ and $i:$ even.}\\ 
4, & \text{otherwise.}
\end{cases}
\end{equation*} 
It can be verified that $c$ defines a conditional $(4,r)$-coloring of $M(C_n)$. Thus $\chi_r (M(C_n)) \leq 4$; hence  $\chi_r (M(C_n))=4$. 
\end{proof}
\newtheorem{thm5}[lem1]{Proposition} 
\begin{thm5}
Let $M(F_n)$ be the middle graph of $F_n$. Then \[\chi_r(M(F_n)) = \left\{  
\begin{array}{l l}
2n+1, & \quad \mbox{if $r \leq 2n$. \textsl{}}\\  
2n+2, & \quad \mbox {if $ r =2n+1$. \textsl{}}\\ 
2n+4, & \quad \mbox {if $ r =\Delta$. \textsl{}}\\ \end{array} \right. \]
\end{thm5}
\begin{proof}
From the definition we have $|V(M(F_n))|=5n+1$. Let $V(M(F_n))= \{v_1,\ldots,v_{5n+1}\}$. We assume that for all $i$ ($1 \leq i \leq n) \; v_{2i-1}$ and $v_{2i}$ represent the edges of $i^{\rm th}$ copy of $K_3$ incident with the center vertex, $v_{2(n+i)}$ and $v_{2(n+i)+1}$ represent the vertices of $i^{\rm th}$ copy of $K_3$ excluding the center vertex,$v_{4n+i+1}$ represents the edge of $i^{\rm th}$ copy of $K_3$ not incident with the center vertex and $v_{2n+1}$ represents the center vertex of $F_n$. It is clear that $\Delta(M(F_n))=2n+2$. We have the following cases\ : \\
\textbf{Case 1:} $r \leq 2n:$ Let $r'=2n$. Since $S=\{v_1,\ldots,v_{2n+1}\}$ is the maximum size clique of $M(F_n), \omega(M(F_n))=2n+1$. We know that $\chi_{r}(G) \geq \omega(G)$, taking $G=M(F_n)$ and $r=r'$, we get $\chi_{r'}(M(F_n))\geq 2n+1$. Now we define the coloring assignment $c \colon V(M(F_n)) \to \{1,\ldots,2n+1 \}$ as follows: 
\begin{equation*}
c(v_i) =
\begin{cases}
i, & \text{if $1 \leq i \leq 2n+1$.}\\
3, & \text{if $i=2n+2$.}\\
4, & \text{if $i=2n+3$.}\\ 
1, & \text{if $2n+4 \leq i \leq 4n+1$ and $(i-2n):$ even.}\\
2, & \text{if $2n+4 \leq i \leq 4n+1$ and $(i-2n):$ odd.}\\
2n+1, & \text{if $4n+2 \leq i \leq 5n+1$.}
\end{cases}
\end{equation*}
It is clear that the total number of colors used in $c$ is $2n+1$. Since $S$ is a clique, $|c(S)|=|S|=2n+1$ and $r'<2n+1$, for all $v \in S$ (C2) is satisfied at $v$. For all $i$  ($1 \leq i \leq n)$  $|c(\{v_{2i-1},v_{2i},v_{2(n+i)},v_{2(n+i)+1},v_{4n+i+1}\})|=|\{v_{2i-1},v_{2i},v_{2(n+i)},v_{2(n+i)+1},v_{4n+i+1}\}|$; therefore the remaining vertices also satisfy (C2). If we take $G$ to be $M(F_n)$ it follows (C3) of lemma 3.2 is satisfied. By lemma 3.2, $\chi_{r'} (M(F_n)) \leq 2n+1$; hence  $\chi_{r'} (M(F_n))= 2n+1$. From [3] we infer $\omega(G) \leq \chi_{r_1}(G) \leq \chi_{r_2}(G)$ if $r_1 \leq r_2$. Taking $G=M(F_n),r_1=r$ and $r_2=r'$, it follows that $\chi_r(M(F_n))= 2n+1$.\\ 
\textbf{Case 2:} $r=2n+1:$ Since $r < \Delta, \; \chi_r(M(F_n)) \geq r+1$. Now we define the coloring assignment $c' \colon V(M(F_n)) \to \{1,\ldots,2n+2 \}$ as follows: 
\begin{equation*}
c'(v_i) =
\begin{cases}
c(v_i)\, \; \text{as defined in case 1}, & \text{if $1 \leq i \leq 4n+1$}\; \text{and}\\
2n+2, & \text{otherwise.}
\end{cases}
\end{equation*}
In $c'$ the \textit{if}-case uses $2n+1$ and the \textit{otherwise}-case uses one new color. Since $|c'(S)|=|S|$ and $c'(S) \cap c'(\{v_{4n+2},\ldots,v_{5n+1}\})= \emptyset$, for all $v \in S$ (C2) is satisfied at $v$. By extending the argument similar to case 1, we can conclude that $c'$ defines a conditional $(2n+2,r)$-coloring of $M(F_n)$. Thus $\chi_{r} (M(F_n)) \leq 2n+2$; hence  $\chi_{r} (M(F_n))= 2n+2$. \\
\textbf{Case 3:} $r=\Delta:\,M(F_n)$ has a Vset-$d2r$ $S_{d2r}=\{v_1,\ldots,v_{2n+3},v_{4n+2}\}$; by lemma 3.1, $\chi_{r} (M(F_n)) \geq |S_{d2r}|= 2n+4$. We now define the coloring assignment $c \colon V(M(F_n)) \to \{1,\ldots,2n+4 \}$ as follows: 
\begin{equation*}
c(v_i) =
\begin{cases}
i, & \text{if $1 \leq i \leq 2n+3$.}\\
2n+3, & \text{if $2n+4 \leq i \leq 4n+1$ and $(i-2n):$ even.}\\
2n+4, & \text{if $i=4n+2$ or both $2n+4 \leq i \leq 4n+1$ and $(i-2n):$ odd.}\\
2n+2, & \text{if $4n+3 \leq i \leq 5n+1$.}
\end{cases}
\end{equation*}  
It is clear that in $c$ the total number of colors used is $2n+4$. For all $i$ ($1 \leq i \leq n)$ we have $|c(S) \cup c(\{v_{2(n+i)},v_{2(n+i)+1},v_{4n+i+1}\})|=|S \cup \{v_{2(n+i)},v_{2(n+i)+1},v_{4n+i+1}\}|$; therefore all the vertices satisfy (C2). With $G=M(F(n))$ we see that (C3) of lemma 3.2 is satisfied. Thus by lemma 3.2,  $\chi_r(M(F_n)) \leq 2n+4$. Hence  $\chi_r (M(F_n))= 2n+4$. 
\end{proof}
\newtheorem{thm6}[lem1]{Proposition}
\begin{thm6}
Let $M(K(n_1,n_2))$ be the middle graph of $K(n_1,n_2)$ and w.l.o.g. assume $n_1 \leq n_2$. Then 
\begin{equation*}
\chi_r (M(K(n_1,n_2))) = 
\begin{cases}
n_2+1, & \text{if $\,r \leq n_2$.}\\
n_2+2, & \text{if $\,r = n_2+1$.} 
\end{cases}
\end{equation*}
\end{thm6}
\begin{proof}
We have $|V(M(K(n_1,n_2)))|=n_1n_2+n$ where $n=n_1+n_2$. Let the vertex set $V(M(K(n_1,n_2)))$ be $\{v_1,\ldots,v_{n+n_1n_2}\}$. We assume that $v_1$ to $v_{n_1}$ represent the vertices of the first partition, $v_{n_1+1}$ to $v_n$ represent the vertices of the second partition and $v_{n+1}$ to $v_{n+n_1n_2}$ represent the edges of $K(n_1,n_2)$. We also assume that for all $i$ ($1 \leq i \leq n_1) \; v_{n+(i-1)n_2+1}$ to $v_{n+in_2}$ represent the edges incident at $v_i$ and for all $j$ ($1 \leq j \leq n_2) \; v_{n+(i-1)n_2+j}$ is incident with $v_i$ and $v_{n_1+j}$. It is clear that $\chi_\Delta (M(K(n_1,n_2))) =n$. We have the following cases: \\ 
\textbf{Case 1:} $r \leq n_2:$ Let $r'=n_2$. Since $\{v_1,v_{n+1},\ldots,v_{n+n_2}\}$ is the maximum size clique of $M(K(n_1,n_2))$, $\omega(M(K(n_1,n_2)))=n_2+1$. We know, $\chi_{r}(G) \geq \omega(G)$; taking $G=M(K(n_1,n_2))$ and $r=r'$, we get $\chi_{r'}(M(K(n_1,n_2)))\geq n_2+1$. We define the coloring assignment $c \colon V(M(K(n_1,n_2)))$ $\to \{1,\ldots,n_2+1 \}$ as follows: 
\begin{equation*}
c(v_i) =
\begin{cases}
n_2+1, & \text{if $1 \leq i \leq n$.}\\
1+(\lfloor(i-1-n)/n_2\rfloor+(i-n)) \mod {n_2}, & \text{otherwise.}
\end{cases}
\end{equation*}
In $c$ the \textit{if}-case uses one and the \textit{otherwise}-case uses $n_2$ new colors. For all $i$ ($1 \leq i \leq n_1) \; |c(N[v_i])|=|c(\{v_i,v_{n+(i-1)n_2+1},\ldots,v_{n+in_2}\})|=|\{v_i,v_{n+(i-1)n_2+1},\ldots,v_{n+in_2}\}|=|N[v_i]|$;  hence for all $v \in$ $V(M(K(n_1,n_2))) \setminus \{v_{n_1+1},\ldots,v_n\}$ (C2) is satisfied at $v$. For all $j$ ($n_1+1 \leq j \leq n) \; |c(N(v_j))|=|c(\{v_{j+in_2}: \forall i: 1 \leq i \leq n_1 \})|=|\{v_{j+in_2}: \forall i: 1 \leq i \leq n_1 \}|=|N(v_j)|$; hence (C2) is satisfied at all $v_j$. Now (C3) of lemma 3.2 is satisfied if we set $G=M(K(n_1,n_2))$. By lemma 3.2, $\chi_{r'} (M(K(n_1,n_2))) \leq n_2+1$; hence  $\chi_{r'} (M(K(n_1,n_2)))= n_2+1$. We know that $\omega(G) \leq \chi_{r_1}(G) \leq \chi_{r_2}(G)$ if $r_1 \leq r_2$. Taking $G=M(K(n_1,n_2)),r_1=r$ and $r_2=r'$, it follows that $\chi_r(M(K(n_1,n_2)))= n_2+1$.\\ 
\textbf{Case 2:} $r=n_2+1:$ Since $r < \Delta,\;$ $\chi_r(M(K(n_1,n_2))) \geq r+1$. We now define the coloring assignment $c' \colon V(M(K(n_1,n_2))) \to \{1,\ldots,n_2+2 \}$ as follows: 
\begin{equation*}
c'(v_i) =
\begin{cases}
n_2+2, & \text{if $1 \leq i \leq n_1$.}\\
c(v_i)\, \text{as defined in case 1}, & \text{otherwise.}
\end{cases} 
\end{equation*}
In $c'$ the \textit{if}-case uses one and the \textit{otherwise}-case uses $n_2+1$ new colors. $c'(\{v_1,\ldots,v_n\}) \cap c'(\{v_{n+1},\ldots,v_{n+n_1n_2}\})$ = $\emptyset,\,c'(\{v_1,\ldots,v_{n_1}\}) \cap c'(\{v_{n_1+1},\ldots,v_n\})=\emptyset$ and for all $i$ ($1 \leq i \leq n) \; |c'(N(v_i))|=|N(v_i)|$; hence (C2) is satisfied at all the vertices. By  setting $G=M(K(n_1,n_2))$ we reason (C3) of lemma 3.2 is satisfied. By lemma 3.2,  $\chi_{r} (M(K(n_1,n_2))) \leq n_2+2$. Hence  $\chi_r (M(K(n_1,n_2))) = n_2+2$.
\end{proof}
\newtheorem{thm9}[lem1]{Proposition}  
\begin{thm9} 
If $G_n$ is the $n$-gear then for any $n \geq 3$ \[\chi_r(G_n) = \left\{ 
\begin{array}{l l}
4, & \quad \mbox{if $r =2$. \textsl{}}\\  
\chi_2(C_{2n})+1, & \quad \mbox{if $ r =3$.\textsl{}}\\ 
min \{r,\Delta \}+1, & \quad \mbox{if $r \geq 4$. \textsl{}}\\  
\end{array} \right. \]
\end{thm9}
\begin{proof}
Let $k = \chi_r(G_n)$. From ~\cite{Lai06} we know $\chi_r(G) \geq min \{r, \Delta\}+1$. Taking $G=G_n$ we have $\chi_r(G_n) \geq min \{r, \Delta\}+1$. Let $k=\chi_r(G_n)$.\\
\textbf{Case 1:} $r=2:$ We have $k \geq 3$. We assume that $k=3$. Let $c \colon V(G_n) \to \{1,2,3 \}$ be a conditional $(3,2)$-coloring where $c(v_i)=i$ for $i=1,2,3$. Then across $v_1, v_2, v_3$ (C1) must be true and in particular (C2) must hold at $v_2$. To satisfy (C1) at $v_3, c(v_4) \neq 3$. We branch into two cases. If $c(v_4)=1$, then by (C2) at $v_4$ we must have $c(v_5) \notin \{1,3\}$. Therefore we must have $c(v_5)=2$. To satisfy (C1), $c(v_0) \notin \{1,2,3\}$. On the other hand if $c(v_4)=2$ then to satisfy (C2) at $v_4$ we must have $c(v_5) \notin \{2,3\}$. Hence $c(v_5)=1$. To satisfy (C2) at $v_3$ while preserving proper coloring, $c(v_0) \notin \{1,2,3\}$. Therefore $k \geq 4$. To show that $k=4$, it suffices to construct a conditional $(4,2)$-coloring of $G_n$. Define $c \colon V(G_n) \to \{1,2,3,4 \}$ as follows:  \\
\textbf{Case 1.1:} $n \equiv 0 \pmod{3}:$ Set $c(v_0)=4$ and
\[c(v_i) = \left\{ 
\begin{array}{l l}
1, & \quad \mbox{if \ $i \bmod 3= 1$. \textsl{}}\\  
2, & \quad \mbox{if \ $i \bmod 3= 2 $.\textsl{}}\\ 
3, & \quad \mbox{if \ $i \bmod 3= 0$. \textsl{}}\\  
\end{array} \right. \]
\textbf{Case 1.2:} $n \equiv 2 \pmod{3}:$ Modify $c$ in case (1.1) by the assignment $c(v_{2n})=2$.\\
\textbf{Case 1.3:} $n \equiv 1 \pmod{3}:$ Modify $c$ in case (1.1) by the assignments $c(v_{2n-4})=2,c(v_{2n-3})=1,c(v_{2n-2})=3,c(v_{2n-1})=2$ and $c(v_{2n})=3$.\\ 
It can be verified that $c$ is a conditional $(4,2)$-coloring of $G_n$. Thus $k \leq 4$, and hence $\chi_2(G_n)= 4$.\\
\textbf{Case 2:} $r=3:$ We have $k \geq 4$. By (C1) $c(v_i) \neq c(v_0)$ and by (C2) at $v_i$, $c(v_{i+1}) \neq c(v_0)$ for all odd $i$ in the range $1 \leq i \leq 2n-1$. Hence $c(v_i) \neq c(v_0)$ for all $i \; (i \neq 0)$.
Since $d(v_i) \leq 3$ for $1 \leq i \leq 2n$ and $r=3$, a conditional $(\chi_3(G_n),3)$-coloring of $G_n$ gives a conditional $(\chi_2(C_{2n}),2)$-coloring of $C_{2n}$. In turn, a conditional $(\chi_2(C_{2n}),2)$-coloring of $C_{2n}$ with an additional color to $v_0$ gives a conditional $(\chi_3(G_n),3)$-coloring of $G_n$.\\ 
\textbf{Case 3:} $r \geq 4 :$ Let $l= \min \{r,\Delta(G_n) \}$. We know $\chi_r(G_n) \geq l+1 $. Since $v_0$ is the only vertex with $d(v_0) \geq r$ and $\chi_r(C_{2n}) \leq l $, conditional $(l,r)$-coloring of $V(G_n) \setminus \{v_0\}$ such that $|c(V_{oddi})| = l $, with $(l+1)$th color assigned to $v_0$ results in a conditional $(l+1,r)$-coloring of $V(G_n)$. Thus $\chi_r(G_n) \leq l+1 $, and hence $\chi_r(G_n) = \min\{r, \Delta\}+1 $.
\end{proof}
\section{Unique conditional colorability}
If $\chi(G)=k$ and every $k$-coloring of $G$ induces the same partition of $V(G)$ then $G$ is called \textit{uniquely k-colorable}. In a similar way we define unique $(k,r)$-colorability of graphs.  \medskip \\
\noindent \textbf{Definition.} If $\chi_r(G)=k$ and every conditional $(k,r)$-coloring of $G$ induces the same partition of $V(G)$, then $G$ will be called \textit{uniquely} $(k,r)$-\emph{colorable}. \medskip \\
With the following propositions we explore further unique $(k,r)$-coloring. 
\newtheorem{pro1}{Proposition}[section] 
\begin{pro1}
If $G$ is uniquely $p$-colorable and $r \leq p-1$ then $\chi_r(G) = p$. 
\end{pro1}
\begin{proof}
Since $G$ is uniquely $p$-colorable, let $c \colon V(G) \to \{1,2,\ldots,p \}$ be the proper coloring of $G$ and w.l.o.g. for $1\leq i \leq p$ let the color class $C_i$ be defined as $C_i=\{v : c(v)=i \}$. For all $u \in V(G)$ if $u \in C_i$ then for all $j \in \{1,2,\ldots,p \}$ there exists a $v \in C_j \;(j \neq i)$ -- this implies that for every $u \in V(G)$, $d(u) \geq p-1 $ and $|c(N(u))|= p-1$. Note that $c$ is also a conditional $(p,r)$-coloring of $G$ because (C2) is also satisfied as for every $u \in V(G),|c(N(u))|= p-1 \geq $ \ min $\{r,d(u)\}$.   
\end{proof}
The definition of conditional $(k,r)$-coloring of $G$ and Proposition 4.1 together imply:
\newtheorem*{cor}{Corollary}  
\begin{cor}
Every uniquely $p$-colorable graph $G$ is also uniquely $(p,p-1)$-colorable.
\end{cor}
\newtheorem{pro2}[pro1]{Proposition}  
\begin{pro2}
For every $k \geq 3$, there exists a uniquely $(3,2)$-colorable graph $G_k$ with $k+2$ vertices.
\end{pro2}
\begin{proof}
We take $G_1$ to be $C_3$. Suppose that $k \geq 1$ and assume  that $G_k$ has been obtained. From $G_k$, we construct $G_{k+1}$ by introducing a new vertex $w$. The vertex and edge sets of $G_{k+1}$ are defined thus:
\begin{align*}
V(G_{k+1}) &= V(G_k) \cup \{w\}, \text{where}\; w \notin V(G_k). \\
E(G_{k+1}) &= E(G_k) \cup \{(u,w),(v,w)\}, \text{where} \ (u,v) \in E(G_k). 
\end{align*}
Evidently,  $|V(G_k)| = k+2$. We need to show that for any $k$, $G_{k}$ is uniquely $(3,2)$-colorable. We prove the result by induction on  $k$. Clearly $G_1$ is uniquely $(3,2)$-colorable. For some $k>1$, assume $G_k$ is uniquely $(3,2)$-colorable. Then $G_{k+1}$ is also uniquely $(3,2)$-colorable because the new vertex $w \in G_{k+1}$ is assigned a third color different from its two neighbors which are adjacent; by the inductive hypothesis the result follows.
\end{proof}
\newtheorem{pro3}[pro1]{Proposition}  
\begin{pro3}
Every path $P_n \;(n \geq 3)$ is uniquely $(3,2)$-colorable.
\end{pro3}
\begin{proof}
From \cite{Lai06} we get $\chi_2(P_n)= 3$ . Define a conditional $(3,2)$-coloring $c \colon V(P_n) \to \{1,2,3 \}$ by    
 \begin{align*}
  c^{-1}(1) = C_1  &=  \{v_i : i \bmod 3= 1 \}, \\
  c^{-1}(2) = C_2  & = \{v_i : i \bmod 3= 2 \}, \\     
  c^{-1}(3) = C_3  & = \{v_i : i \bmod 3= 0 \},
 \end{align*} 
where, $C_1$,$C_2$ and $C_3$ are the color classes. In the conditional $(3,2)$-coloring of $P_n$, for any two vertices  $v_i,v_j \; (i \ne j)$ in $V(P_n)$ if $|i-j| \bmod 3 = 0$ then $v_i$ and $v_j$ must be colored same; otherwise either (C1) will be violated at $v_{min\{i,j\}+1}$ or (C2) will be violated at $v_{max\{i,j\}-1}$. Since $c$ is the only coloring wherein for all $v_i,v_j \in V(P_n)$, $c(v_i) = c(v_j)$, if $|i-j| \bmod 3= 0$,  $P_n$ is uniquely $(3,2)$-colorable.    
\end{proof}
\newtheorem{pro4}[pro1]{Proposition}  
\begin{pro4}
If $T$ ($\neq P_n$) is a rooted tree with $n$ verices and $k = \chi_r(T)$ then $T$ is not uniquely $(k,r)$-colorable unless $k=n$. 
\end{pro4}
\begin{proof}
Let the root of $T$ be a vertex of degree $\Delta(T)$. Let $c \colon V(T) \to \{1,2,\ldots,k \}$ be a conditional $(k,r)$-coloring of $T$. We show that a new conditional $(k,r)$-coloring of $T$ can be obtained based on $c$ if $k \neq n$. We know that every tree $T$ has  at least two vertices say, $u,v$ with degree one. Let $p(v)$ denote the parent of $v$. We have the following cases:\\ 
\textbf{Case 1:} $p(u)= p(v):$ \\ 
\textbf{Case 1.1:} $c(u) = c(v):$ Since $r \geq 2$ assigning one of the colors from the set  $c(N_T(p(u))) \setminus \{c(v)\}$ to $u$ results in a new conditional $(k,r)$-coloring. \\
\textbf{Case 1.2:} $c(u) \neq c(v):$ \\
\textbf{Case 1.2.1:} There exists a vertex $w \in V(T)$ such that $c(w) = c(u)$ or $c(w)=c(v):$ Interchanging the colors of $u$ and $v$ results in a different induced partition of $V(T)$.\\
\textbf{Case 1.2.2:} There doesn't exist a vertex $w \in V(T)$ such that $c(w) = c(u)$ or $c(w)=c(v):$\\
\textbf{Case 1.2.2.1:} $r \geq \Delta :$ If $n \neq \Delta+1$ then there exists a vertex $w' \in N_T(p(u))$ such that $d(w') \geq 2$; then interchanging the colors of $u$ and $w'$ results in a different induced partition of $V(T)$ because the subtree rooted at $w'$ doesn't contain any vertex colored $c(u)$. If $n = \Delta+1$ then $k=n$.\\
\textbf{Case 1.2.2.2:} $r < \Delta :$ There must exist at least two vertices $w_1, w_2 \in N_T(p(u))$ such that $c(w_1) = c(w_2)$. Interchanging the colors of $u$ and $w_1$ results in a different induced partition of $V(T)$ because the subtree rooted at $w_1$ doesn't contain any vertex colored $c(u)$. \\
\textbf{Case 2:} $p(u)\neq  p(v):$  \\
\textbf{Case 2.1:} $d(p(u)) < $ min $\{r,\Delta(T)\}$ or $d(p(v)) < $ min $\{r,\Delta(T)\} :$ Assigning to $u$ any color in the set $c(V(T)) \setminus c(N_T[p(u)])$ if $d(p(u)) < $ min $\{r,\Delta(T)\}$ or to $v$ any color in the set $c(V(T)) \setminus c(N_T[p(v)])$ if $d(p(v)) < $ min $\{r,\Delta(T)\}$ gives a new conditional $(k,r)$-coloring. \\
\textbf{Case 2.2:}  $d(p(u)) > $ min $\{r,\Delta(T)\}$ or $d(p(v)) > $ min $\{r,\Delta(T)\} :$ If $d(p(u)) > $ min $\{r,\Delta(T)\}$ there must exist a vertex $u' \in N_T(p(u))$ such that $c(u) = c(u')$ or two vertices $u_1,u_2 \in N_T(p(u))$ such that $c(u_1) = c(u_2) \neq c(u)$. Assigning to $u$ any of the color in the set $c(V(T)) \setminus \{c(u'),c(p(u))\}$ or  interchanging the colors of $u$ and $u_1$ and the colors $c(u)$, $c(u_1)$ in the subtree rooted at $u_1$ gives a new conditional $(k,r)$-coloring in the former and latter cases respectively. The case $d(p(v)) > $ min $\{r,\Delta(T)\}$ is similar.\\ 
\textbf{Case 2.3:}  $d(p(u)) = d(p(v)) = $ min $\{r,\Delta(T)\} :$ \\
\textbf{Case 2.3.1:} min $\{r,\Delta(T)\} > 2 :$ There exists a $w \in V(T)$ such that $c(u) \neq c(w)$ and $p(u)=  p(w)$. Making the color of $u$ as $c(w)$ and in the subtree rooted at $w$,  swapping the colors $c(u)$ and $c(w)$ gives a new conditional $(k,r)$-coloring. \\
\textbf{Case 2.3.2:} min $\{r,\Delta(T)\} = 2 :$ Since $T \neq P_n$ so $\Delta(T) > 2$ and $r=2$. There exists an ancestor of $u$ and $v$ with degree $\geq 2 $ because the root of $T$ has maximum degree. Let $w$ with $d(w) \geq 2$ be the closest ancestor of $u$ and $v$. Then there exists two children of $w$ namely $w_1$ and $w_2$ which are ancestors of $u$, $v$ respectively. If $w$ is the root of $T$ and $c(w_1) \neq c(w_2)$, interchanging the colors $c(w_1)$ and $c(w_2)$ in the subtree rooted at $w_1$ gives a new conditional $(k,r)$-coloring. If $w$ is not the root of $T$ and $c(w_1) \neq c(w_2)$, interchanging the colors $c(w_1)$ and $c(w_2)$ in the subtree rooted at $w$ gives a new conditional $(k,r)$-coloring. Otherwise (i.e., if $c(w_1) = c(w_2)$), there exists a $w_3 \in N_T(w)$ such that $c(w_3) \neq c(w_1)$ and interchanging the colors $c(w_1)$ and $c(w_2)$ in the subtree rooted at $w_1$ gives a new conditional $(k,r)$-coloring.
\end{proof}


\begin{thebibliography}{99}
\bibitem{Bal}R. Balakrishnan and K. Renganathan, {\it A Textbook of Graph Theory}, Springer, 2000.
\bibitem{Galin09} J. A. Gallian, {\em A dynamic survey of graph labeling}, Elect. J. Combin. 16 (DS6) (2009).
\bibitem{Gol80} M. C. Golumbic, {\em Algorithmic graph theory and perfect graphs}, Elsevier, North Holland, 2004.
\bibitem{Har69}	F. Harary, {\em Graph theory}, Addison-Wesley, MA, 1969.
\bibitem{Lai06}H. J. Lai, J. Lin, B. Montgomery, T. Shui and S. Fan, {\em Conditional colorings of graphs}, Discr. Math. 306 (2006), pp. 1997-2004.
\bibitem{Li09}X. Li, X. Yao, W. Zhou and H. Broersma, {\em Complexity of conditional colorability of graphs}, Appl. Math. Lett. 22 (2009), pp. 320-324.
\bibitem{Mich83}D. Michalak, {\em On middle and total graphs with coarseness number equal 1}, Lect. Notes in Math. 1018: Graph Theory, Springer, (1983), 139-150. 
\bibitem{West03} D. B. West, {\em Introduction to graph theory}, Prentice-Hall, 2003.
\end{thebibliography}
\end{document}